\documentclass[conference,10pt]{IEEEtran}
\usepackage{graphicx}
\usepackage{amsfonts,color}
\usepackage{cite}
\usepackage{amssymb}
\usepackage{epsfig}
\usepackage{latexsym}
\usepackage[T1]{fontenc}
\usepackage{epstopdf}
\usepackage{amsmath}
\usepackage{ amsthm, amssymb}
\usepackage{verbatim}
\usepackage{algorithmic}
\usepackage[algo2e]{algorithm2e}
\usepackage{algorithm}
\usepackage{subcaption}
\usepackage[utf8]{inputenc}
\usepackage{soul}

\newtheorem{theorem}{Theorem}

\newtheorem{remark}{Remark}
\definecolor{orange}{RGB}{255,107,0}

\begin{document}
\title{Unsupervised Learning for
Pilot-free  Transmission in 3GPP
MIMO Systems}
	\author{\IEEEauthorblockN{Omar M. Sleem$^{\dagger *}$, Mohamed Salah Ibrahim$^*$, Akshay Malhotra$^*$, Mihaela Beluri$^*$, and Philip Pietraski$^*$}
		\IEEEauthorblockA{$^{\dagger}$ Department of Electrical and Computer Engineering, Pennsylvania State University \\ 
$^*$ InterDigital Communications\\
			                  Email: \{\tt oms@psu.edu, firstname.lastname@interdigital.com\}
		                 }
	           }
%\ninept
% \setlength{\belowdisplayskip}{2pt}
% \setlength{\belowdisplayshortskip}{2pt}
% \setlength{\abovedisplayskip}{2pt}
% \setlength{\abovedisplayshortskip}{2pt}
% \setlength{\textfloatsep}{5pt}
\maketitle
\begin{abstract}
%\begin{sloppypar}
Reference signals overhead reduction has recently evolved as an effective solution for improving the system spectral efficiency. This paper introduces a new downlink data structure that is free from demodulation reference signals (DM-RS), and hence does not require any channel estimation at the receiver. The new proposed data transmission structure involves a simple repetition step of part of the user data across the different sub-bands. Exploiting the repetition structure at the user side, it is shown that reliable recovery is possible via  canonical correlation analysis. This paper also proposes two effective mechanisms for boosting the CCA performance in OFDM systems; one for repetition pattern selection and another to deal with the severe frequency selectivity issues. The proposed approach exhibits favorable complexity-performance tradeoff, rendering it appealing for practical implementation. Numerical results, using a 3GPP link-level testbench, demonstrate the superiority of the proposed approach relative to the state-of-the-art methods.  
%\end{sloppypar}
\end{abstract}

\section{Introduction}
%\begin{sloppypar}
% General paragraph on current challenges
The ever-growing number of consumer wireless devices, together with the newly emerged data-hungry applications, have introduced unprecedented spectrum usage challenges \cite{saad2019vision}. Even with cutting-edge technologies, as massive MIMO \cite{lu2014overview} and mmWave communications \cite{hong2021role}, currently in place, there are still several challenges for further improvements in the end-user data rates, cost efficiency, and communication latency. These challenges raised the need for developing new techniques for enhancing the system performance.    

%% RS OH reduction
Reference signals (RS) overhead reduction approaches have recently gained significant interest as means for improving the system performance \cite{ahmed2019overhead}. RS are located in some specific resource elements (REs) in the time-frequency grid. The 5G NR specifications define several types of RS that are transmitted in different ways and intended to be used for different purposes \cite{3gpp38211}, \cite{3gpp38213}, \cite{3gpp38214}. For instance, CSI-RS are downlink RS dedicated for downlink channel estimation, while demodulation reference signals (DM-RS) are intended for estimating the effective channel (precoded channel) as part of coherent demodulation at the receiver. 

%% RS-based estimation and Equalization drawbacks  
Decoding the signal of interest at the receivers generally requires a two-step solution; first estimating the channel followed by designing an equalizer as a function of the estimated channel. In cellular systems, RS-based channel estimation may require a large number of RS to ensure acceptable estimation performance; a requirement that hurts the data transmission rate since more REs are reserved for RS as opposed to data. The situation is further complicated in the multiuser MIMO setup where the RS need to be orthogonal across the different UEs, which is a difficult proposition on its own, as it requires tight coordination between the different base stations. Once the channel is estimated, an equalizer is designed to recover signal of interest. Several equalization approaches with various levels of complexity  can be used to decode the signal of interest \cite{madhow1994mmse,wiesel2005semidefinite}. However, such approaches require accurate channel estimate to obtain acceptable channel estimation performance.

%% Contribution
This paper seeks to address the aforementioned limitations of the pilot-based approaches via proposing a a new DM-RS free data structure that employs a simple repetition protocol on a portion of the user data in the time-frequency grid. The repetition structure will then be utilized at the receiver side to create two data views -- assuming that the repetition pattern is known at the receiver and is different across layers. Applying canonical correlation analysis (CCA) \cite{hotelling1936relations} on the two constructed data views, for each layer, will recover the desired signal at the UE. We introduce a strategy for choosing the repetition type that relies on the channel parameters. In addition, we propose a solution to boost the CCA performance in scenarios where frequency selectivity is severe. Finally, we showed through simulations that the proposed method outperforms the state of the art method, while having lower complexity.

%% CCA literature
CCA is a statistical learning tool that has a wide variety of applications in signal processing and wireless communications. This includes, but not limited to, direction-of-arrival estimation~\cite{wu1994music}, equalization \cite{dogandzic2002finite}, spectrum sharing \cite{SalahUnderlay22}, blind source separation~\cite{bertrand2015distributed},  cell-edge user detection~\cite{salahtwc}, to list a few.

\section{System Model and Problem Statement}\label{Sec: probdesc}
Consider a downlink (DL) transmission in a 5G NR network where a single base station (gNB) sends data to a single user equipment (UE) through a physical DL shared channel (PDSCH). The gNB is equipped with $N_{\mathrm{t}}$ antennas and serves a UE of $N_{\mathrm{r}}$ antennas. For sub-band $n$, the gNB transmits $L$ data streams; each of length $N_{\mathrm{data}}$ and represented by the matrix $\mathbf{X}^{(n)}\in \mathbb{C}^{N_{\mathrm{data}}\times L}$, {for $n \in [N] := \{1,\cdots,N\}$}, and $N$ is the number of sub-bands. The received base-band signal for sub-band $n$ can then be described as follows,
%%\mathbf{Y}^{(n)}=\mathbf{H}^{(n)}\mathbf{F}^{(n)}\mathbf{X}^{(k)^\top}+\mathbf{W}^{(n)},
\begin{equation}
    {\bf Y}^{(n)} = {\bf H}^{(n)}\sum\limits_{\ell = 1}^{L} \sqrt{\alpha_\ell}{\bf f}^{(n)}_\ell{\bf x}^{\top(n)}_\ell + {\bf W}^{(n)}
\end{equation}
where $\mathbf{Y}^{(n)}\in\mathbb{C}^{N_{\mathrm{r}}\times N_{\mathrm{data}}}$ is the $N_r$-dimensional received signal at the UE, $\mathbf{H}^{(n)}\in\mathbb{C}^{N_{\mathrm{r}}\times N_{\mathrm{t}}}$ is the DL channel response matrix associated with the $n$-th sub-band, and $\mathbf{W}^{(n)}\in\mathbb{C}^{N_{\mathrm{r}}\times N_{\mathrm{data}}}$ contains independent and identically distributed (i.i.d) complex Gaussian entries of zero mean and variance $\sigma^{2}$ each. The term $\alpha_\ell$ represents the power allocated to the $\ell$-th layer, where unless stated otherwise, $\alpha_\ell$ is chosen to be $1/\sqrt{L}$, $\forall \ell$. In order to support multi-stream transmission, the gNB precodes the $\ell$-th stream data symbols  $\mathbf{x}_{\ell}^{(n)} \in \mathbb{C}^{N_{\mathrm{data}}}$ using the precoder $\mathbf{f}_\ell^{(n)}\in\mathbb{C}^{N_{\mathrm{t}}}$, {where $\|{\bf f}_\ell\|_2^2 = 1$}.  For the SU-MIMO case considered herein, we assume SVD based precoding where the precoder ${\bf f}^{(n)}$ is the $\ell$-th dominant right singular vector of the matrix ${\bf{H}}^{(n)}$.  

Throughout this work, the DL effective channel matrix $\bar{\mathbf{h}}_{\ell}^{(n)}\stackrel{\Delta}{=}\mathbf{H}^{(n)}\mathbf{f}_\ell^{(n)}$ is assumed to be unknown at the UE. We assume a wide band (WB) precoding of the entire resource grid (RG). This yields a single precoder for each layer, $\mathbf{f}^{(n)}_\ell = \mathbf{f}_\ell, \forall n \in [N]$, that is chosen to be the $\ell$-th dominant right singular vector of the average of all channel matrices across the RG, and hence, we suppress the superscript dependence of $\mathbf{f}_\ell^{(n)}$. 

\begin{figure}[t]
\begin{center}
\includegraphics[scale=.35]{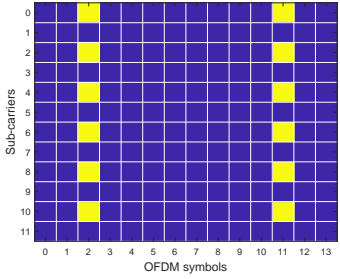}  
\caption{Traditional wideband data channel structure with data REs colored in blue while DM-RS REs are colored in yellow.} \label{DM-RS}
\end{center}
\end{figure}

%% need to highlight the DM-RS shortcomings here

In 3GPP 5G NR, data transmission is accompanied with a set of reference signals, such as  DM-RS. DM-RS for PDSCH are intended for the estimation of $\bar{\mathbf{h}}_{\ell}^{(n)}$ as part of coherent demodulation. DM-RS based channel estimation and equalization is a computationally complex process. The channel estimation accuracy depends on the DM-RS density; increasing the DM-RS density may improve the channel estimation accuracy, but it could potentially hurt the data rate as fewer REs are then used for data transmission. Moreover, for the multi-user (MU) transmission case, the DM-RS sequences associated with the co-scheduled UEs need to be orthogonal -- a constraint that is difficult to achieve in the multi-cell networks. We will next show how to overcome the DM-RS short-comings mentioned by designing a machine leaning (ML) based detection method that tends to recover the transmitted data symbols in an unsupervised manner, and at much lower complexity relative to the DM-RS approach.

\begin{figure}[t]
     \centering
     \begin{subfigure}[b]{0.24\textwidth}
         \centering
         \includegraphics[width=1\textwidth]{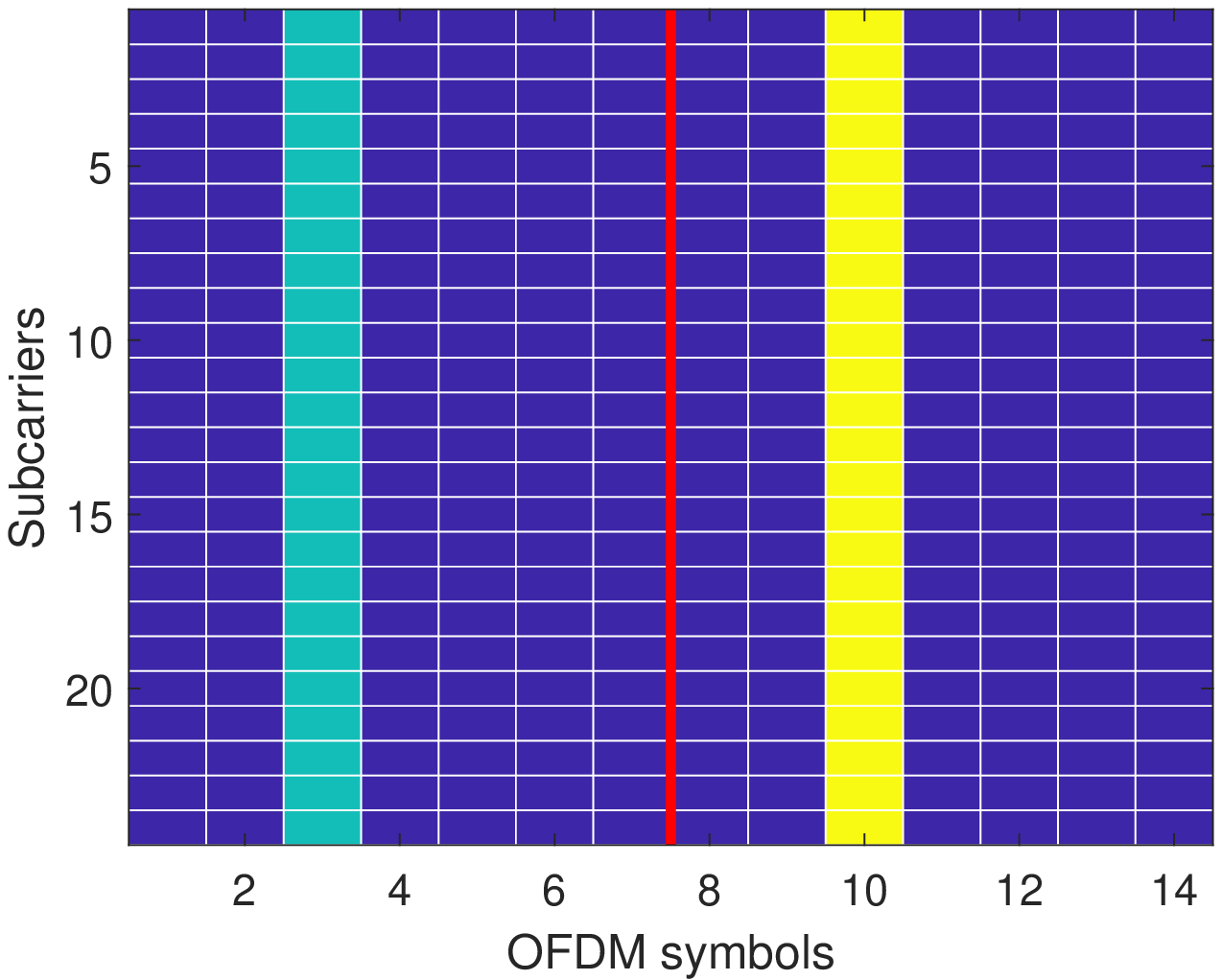}
         \caption{Pattern 1.}
         \label{fig:CCApatt1}
     \end{subfigure}
     \begin{subfigure}[b]{0.24\textwidth}
         \centering
         \includegraphics[width=1\textwidth]{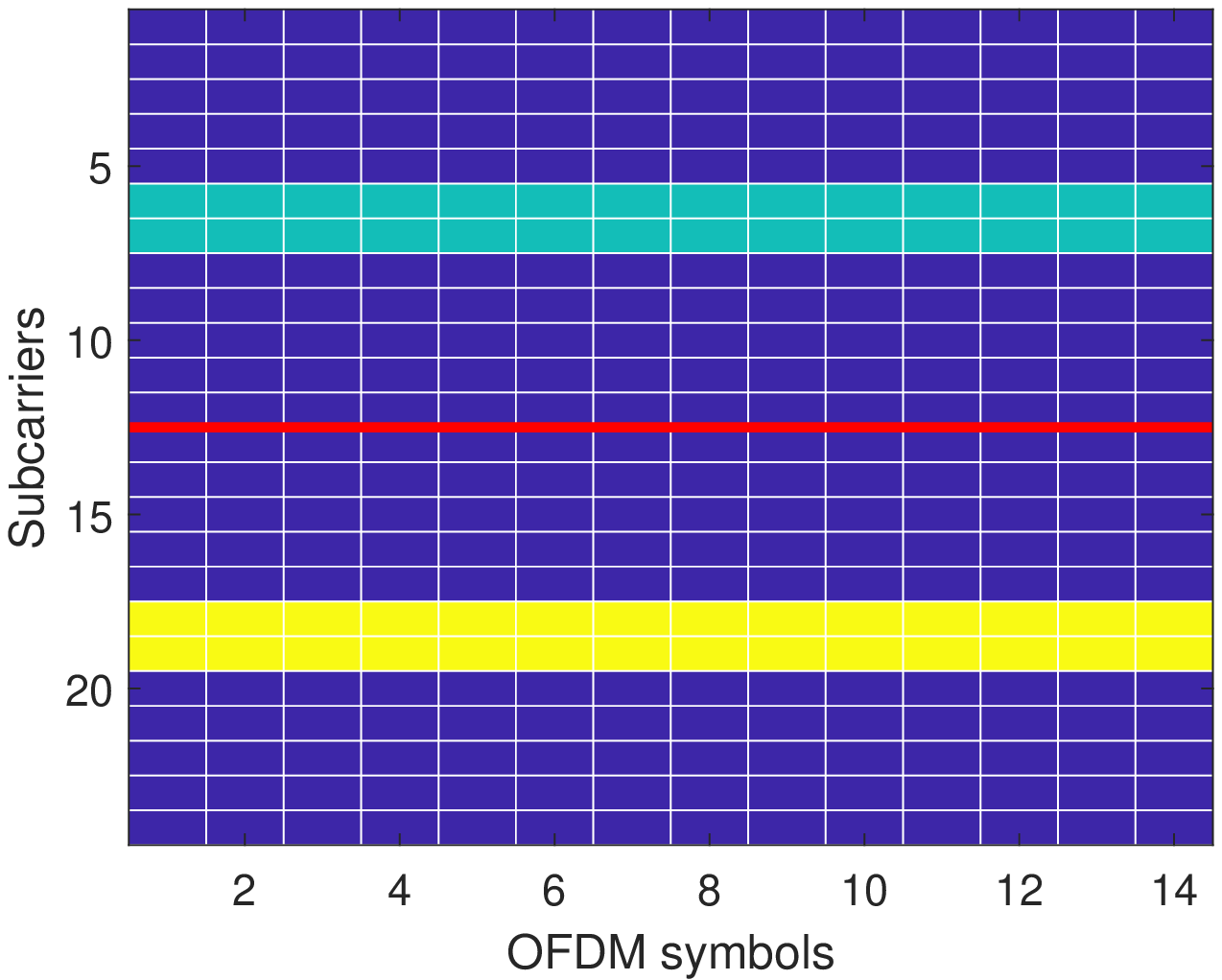}
         \caption{Pattern 2.}
         \label{fig:CCApatt3}
     \end{subfigure}
        \caption{Proposed data structure with data REs colored in blue and light blue and reserved REs colored in yellow.}
        \label{fig:CCApatts}
\end{figure}

\section{Proposed Method}
In this section, we explain how CCA can be exploited in order to solve the problem described in Section \ref{Sec: probdesc}. Instead of using the frame structure in Fig. \ref{DM-RS} that contains reserved REs for DM-RS symbols (colored in yellow), we propose a new DM-RS-free data transmission structure. The new  structure requires repeating a few data symbols in the time-frequency grid, as depicted in Fig. 2. The repetition may be employed either in time or frequency, where the different repetition patterns can yield different performance, as will be explained in Section \ref{Sec:CCAPatterns}. The repetition structure will then be utilized at the UE side to derive the combiners that will be used to decode the PDSCH data in the repeated locations and in the  neighbourhood of the repeated symbols, as will be shown later. To explain how the proposed method works, let ${\bf x}_{\mathrm{c}\ell} \in \mathbb{C}^{\Bar{N}}$ denote the common/repeated signal associated with the $\ell$-th layer within a part of the sub-band, where $\Bar{N}$ represents the length of the common signal and $\Bar{N} << N_{\mathrm{data}}$. 
%While the DM-RS framework requires the DM-RS sequences to be orthogonal in the multilayer transmission case, our approach merely requires the repetition patterns associated with the different layers to be not identical -- a very mild assumption that can be easily satisfied as we will see in in Section \ref{Sec:CCAPatterns}. 
Towards this end, the baseband equivalent model of the received signal at each of those two separate regions using the repetition pattern of the $\ell$-th layer, can be expressed as,
\begin{equation} \label{CCAViews} 
  {\bf Y}^{(i)}_{\ell} =  \bar{\bf h}_\ell^{(i)}{\bf x}^{\top}_{\mathrm{c}\ell} + \sum\limits_{j \neq \ell}^{L}  \bar{\bf h}_j^{(i)}{\bf x}^{\top}_{ij}  +\overline{\mathbf{W}}^{(i)}_{\ell} 
\end{equation}
where $i=1,2$ refers to the region/view index, and ${\bf x}_{ij}$ is the signal associated with the received signal of the $j$-th layer in the $i$-th region. The power allocation terms are absorbed in the respective channel vectors. It should be noted that, under the assumption that the repetition patterns are unique, the two views in \eqref{CCAViews} share one common signal associated with the layer which repetition pattern is used -- all the other layers will be randomly permuted, and hence, having different signals in both views. We will now show how CCA can be utilized at the UE to recover the desired signal. 
%Note that the symbols in each region are assumed to have a constant effective channel (we will explain in Section \ref{Sec:CCAPatterns} how the patterns are chosen to satisfy such an assumption).

CCA is a widely-used machine learning tool that seeks to find two linear combinations of two given random vectors such that the resulting pair of random variables are maximally correlated. In a recent work \cite{salahtwc}, the authors came up with a new algebraic interpretation of CCA as a method that can identify a common subspace between two given data views, under a linear generative model. It was shown that, if two signals views have a common/shared component in addition to individual components across each view, then applying CCA to those two views will recover the common component up to a global complex scaling ambiguity, where the scaling ambiguity can be resolved via correlation a few pilots symbols with the corresponding received one. It can be easily seen that such an interpretation is applicable to the model in \eqref{CCAViews}, where the two data views in \eqref{CCAViews} only share the signal of the layer which repetition pattern is used. Considering the following MAXVAR formulation of CCA \cite{hardoon2004canonical},
\begin{subequations}\label{MAXVAR}
	\begin{align}
	&\underset{{\bf g}_{\ell},{\bf q}_{1\ell},{\bf q}_{2 \ell}}{\min}~  \sum_{j = 1}^{2}\| {\bf Y}^{(i)\top}{\bf q}_{i\ell} - {\bf g}_{\ell}\|^2_2,\\
	& \text{s.t.} \quad \|{\bf g}_{\ell}\|^2_2 = 1.
	\end{align}
\end{subequations}
The optimal canonical vectors (referred to as combiners in the considered problem), ${\bf q}^\star_{1\ell}$ and ${\bf q}^\star_{2 \ell}$, can be obtained through solving an eigen-decomposition problem. 
To clarify how the results of \cite{salahtwc} can be mapped to the considered problem herein, let us define the matrix ${\bf X}_{i\ell} := [{\bf x}_{\mathrm{c}\ell}, {\bf X}_{i}] \in \mathbb{C}^{\Bar{N} \times L}$,  where the columns of  ${\bf X}_{i} \in \mathbb{C}^{\Bar{N} \times (L-1)}$ represent the signal associated with the other/interfering $L-1$ layers. Similarly, let the matrix $\overline{\bf H}_i  \in \mathbb{C}^{N_r \times L}$ hold the corresponding effective channel vectors. Towards this end, we have the following result.
\begin{theorem}
		In the noiseless case, if the matrices ${\bf X}_{i\ell}$  and $\overline{\bf H}_i$, for $i=1,2$ and $\ell \in [L]\stackrel{\Delta}{=}\{1,\dots L\}$, are full column rank, then the optimal solution ${\bf g}_\ell^\star$ of \eqref{MAXVAR} is given as ${\bf g}_\ell^\star = \gamma_\ell {\bf x}_{\mathrm{c}\ell}$, where $\gamma_\ell \in \mathbb{C}$ is a complex phase ambiguity.
\end{theorem}
\begin{proof}
	The proof follows from Theorem 1 in~\cite{salahtwc}.
\end{proof}
\begin{remark}
 The full rank condition on the matrix ${\bf X}_{i\ell}$ requires the CCA view length $\Bar{N}$ to be greater than or equal to the number of layers $L$ and the columns of ${\bf X}_{i\ell}$ to be linearly independent. Since the transmissions associated with the different layers are independent, then such requirements can be easily satisfied with modest $\Bar{N}$. On the other hand, the full rank condition on the effective channel matrix $\overline{\bf H}_i$ can be satisfied if $i)$ the number of receive antennas is greater than or equal to the number of layers and $ii)$ the columns of $\overline{\bf H}_i$ are linearly independent. Both conditions are satisfied with probability one since $i)$ the number of transmitted streams $L$ is always upper-bounded by $N_r$ and $ii)$ the columns of $\overline{\bf H}_i$ are in fact orthogonal (because of the SVD-based precoding). 
 \end{remark}
 
 %% Omar will squeeze the below subsections in one column and half
\subsection{Pattern Design}\label{Sec:CCAPatterns}
For the $\ell$-th layer, we consider a RG of $N_{\mathrm{RB}}$ resource blocks (RB)s. Each RB is composed of REs that are divided among 12 sub-carriers vertically and 14 OFDM symbols horizontally. $N_{\mathrm{data}}$ REs in the RG are loaded with data symbols from the transmit vector $\mathbf{x}_{\ell{}}$, such that each symbol corresponds to a RE, while $N_{\mathrm{res}}$ REs are reserved for control signals. 

REs in the $\ell$-th layer RG are located using matrix linear indexing. Assuming that the (sub-carrier, slot) location of an RE in the RG is $(m,n)$ and both $m$ and $n$ start from 0, the corresponding index can be found as $i_{m,n}=12 n  N_{\mathrm{RB}}+m+1$. We define $\mathcal{I}_{\mathrm{RE}}$ as the set of all indices in a RG, where for $\ell \in [L]$, $\mathcal{I}_{\mathrm{RE}}=\mathcal{I}_{\mathrm{data}}^{\ell}\cup\mathcal{I}_{\mathrm{res}}^{\ell}$, such that $\mathcal{I}_{\mathrm{data}}^{\ell}$ and $\mathcal{I}_{\mathrm{res}}^{\ell}$ are two disjoint sets that correspond to the indices of data and reserved REs for layer $\ell$,  respectively. %It is important to note that $\mathcal{I}_{\mathrm{RE}}$ (along with $\mathcal{I}_{\mathrm{data}}$ and $\mathcal{I}_{\mathrm{res}}$) is a function of the SB index $k$ within the RG, however, we dropped that index dependency for notation simplicity. 

Let the indices set $\mathcal{I}_{\mathrm{s}}^{\ell}\subset \mathcal{I}_{\mathrm{data}}^{\ell}$ such that $|\mathcal{I}_{\mathrm{s}}^{\ell}|=N_{\ell}<|\mathcal{I}_{\mathrm{data}}^{\ell}|=N_{\mathrm{data}}$, where $|.|$ is the cardinality of a set. For ease of notations, we assume equal view length for all layers, i.e., $N_{\ell}=\Bar{N}, \forall \ell$. In this case, the common signal $\mathbf{x}_{\mathrm{c}\ell}\in\mathbb{C}^{\Bar{N}}$, is  the vector of symbols formed by the data in the REs located with the index set $\mathcal{I}_{\mathrm{s}}^{\ell}$. We also let the index set $\mathcal{I}_{\mathrm{d}}^{\ell}\subseteq \mathcal{I}_{\mathrm{res}}^{\ell}$, such that $|\mathcal{I}_{\mathrm{d}}^{\ell}|=|\mathcal{I}_{\mathrm{s}}^{\ell}|=\Bar{N}$. We refer to $\mathcal{I}_{\mathrm{s}}^{\ell}$, $\mathcal{I}_{\mathrm{d}}^{\ell}$ and the tuple $\mathcal{P}^{\ell}\stackrel{\Delta}{=}(\mathcal{I}_{\mathrm{s}}^{\ell}, \mathcal{I}_{\mathrm{d}}^{\ell})$ by source, destination indices sets and a CCA pattern,  respectively. 

\subsection{Data repetition} \label{Datarep}
At time slot $t$, the UE generates the source indices set for the $\ell$-th layer, $\mathcal{I}_{\mathrm{s}}^{\ell}$, such that the symbols of the corresponding CCA signal $\mathbf{x}_{\mathrm{c}\ell}$ lie within the time-bandwidth coherence (TBC) block $\mathfrak{B}_{1}^{\ell}$. Similarly, the destination indices set $\mathcal{I}_{\mathrm{d}}^{\ell}$ is generated such that $|\mathcal{I}_{\mathrm{d}}^{\ell}|=|\mathcal{I}_{\mathrm{s}}^{\ell}|=\Bar{N}$ with REs in the TBC block $\mathfrak{B}_{2}^{\ell}$. The CCA pattern $\mathcal{P}^{\ell}=(\mathcal{I}_{\mathrm{s}}^{\ell},\mathcal{I}_{\mathrm{d}}^{\ell})$ is then signaled to the gNB to be used in the next slot transmission. At time slot $t+1$, the gNB repeats the CCA signal $\mathbf{x}_{\mathrm{c}\ell}$, constructed from the indices in $\mathcal{I}_{\mathrm{s}}^{\ell}$, in the locations defined by $\mathcal{I}_{\mathrm{d}}^{\ell}$. 

%The dissimilarity in the channel's impulse response of the two TBC blocks $\mathfrak{B}_{1}^{\ell}$ and $\mathfrak{B}_{2}^{\ell}$ allows for the engendering of two different views of the same entity, $\mathbf{x}_{\mathrm{c}\ell}$, at the UE side. 
The UE constructs the two views $\mathbf{Y}_{\ell}^{(i)}\in\mathbb{C}^{N_{\mathrm{r}}\times \Bar{N}}, i=\{1,2\}$, of the CCA signal $\mathbf{x}_{\mathrm{c}\ell}$ defined in \eqref{CCAViews}.
%\begin{equation}
%    \mathbf{Y}_{i}=\mathbf{H}_{i}\mathbf{f}\mathbf{x}^{\top}_{\mathrm{c}}+\Bar{\mathbf{W}}, 
%\end{equation}
 %where $\mathbf{h}_{\mathrm{eff}}^{(i\ell)}$ is the DL channel impulse response matrix that corresponds to TBC block $\mathfrak{B}_{i}^{\ell}$ for layer $\ell$ and $\Bar{\mathbf{W}}_{i\ell}\in\mathbb{C}^{N_{\mathrm{r}}\times \Bar{N}}$ contains i.i.d complex Gaussian entries of zero mean and variance $\sigma^{2}$. 
 Given the two views, $\mathbf{Y}_{\ell}^{(i)}, i=\{1,2\}$, the UE then solves the problem in \eqref{MAXVAR} to find the combiners $\mathbf{q}^{*}_{i\ell}\in\mathbb{C}^{N_{\mathrm{r}}}, i=\{1,2\}$. $\mathbf{q}_{i\ell}^{*}$ is then used to combine TBC block $\mathfrak{B}_{i}^{\ell}$ and the $N_{\mathrm{\mathfrak{B}}_{i}\ell}$ REs in its vicinity such that $N_{\mathrm{\mathfrak{B}}_{1}\ell}+N_{\mathrm{\mathfrak{B}}_{2}\ell}=N_{\mathrm{data}}-\Bar{N}$.

% A TBC block's shape depends on the channel impulse response. For flat fading and time varying channels, TBC blocks exhibit a tall matrix shape, where the blocks are repeated in time to capture the possible time variation. Conversely, for frequency selective and time invariant channels, the blocks depict a fat matrix shape, where the repetition is done over the frequency band. Figure \ref{Rep_schemes} provides two examples for time and frequency repetition schemes for a RG composed of one RB. The red line defines the vicinity regions for both $\mathfrak{B}_{1}$ and $\mathfrak{B}_{2}$. When repeating in time, we assume that the REs in slots $0\rightarrow 6$, within layer $\ell$, are combined using $\mathbf{q}_{1\ell}^{*}$, while $\mathbf{q}_{2\ell}^{*}$ is used to combined the rest (figure \ref{fig:Reptime}), i.e., REs in slots $7 \rightarrow 13$. For frequency repetition, the entire SB bandwidth (BW) is divided to two parts, where $\mathbf{q}_{i\ell}$ is used to combine REs part $i, i\in \{1,2\}$, (figure \ref{fig:Repfrequency}).

\subsection{Sub-gridding}
As mentioned in the previous section, the combiner $\mathbf{q}_{i\ell}^{*}$ is used to combine the REs in the vicinity region of $\mathfrak{B}_{i}^{\ell}$. However, this only serves as an approximation for the optimal combiners required for those REs. In order to reduce that approximation error, we introduce the RG sub-grids (SG)s solution. The RG is divided into equal sized non-overlapping SGs, each of size $N_{\mathrm{BSG}}$ RBs, where SG $j$ contains its own CCA pattern $\mathcal{P}_{j}^{\ell}=(\mathcal{I}_{s,j}^{\ell},\mathcal{I}_{d,j}^{\ell})$, and signal $\mathbf{x}_{\mathrm{c}\ell,j}$. The pattern $\mathcal{P}_{j}^{\ell}$ is then used to construct the two views, $\mathbf{Y}_{\ell,j}^{(i)}, i\in\{1,2\}$, of the CCA signal $\mathbf{x}_{\mathrm{c}\ell,j}$ within SG $j$. Given the per SG views, the CCA problem is solved to find the optimal combiners $\mathbf{q}_{i\ell,j}^{*}, i\in\{1,2\}$. Combiner $\mathbf{q}_{i\ell,j}^{*}$ is then used to combine the REs in its vicinity in SG $j$.  

We define the set $\mathcal{I}_{s,j}^{\ell,m}$ $(\mathcal{I}_{d,j}^{\ell,m})$ to be the set of source (destination) REs in RB $m$ of SG $j$ such that $\bigcup_{m\in [N_{\mathrm{BSG}}]}\mathcal{I}_{x,j}^{\ell,m}=\mathcal{I}_{x,j}^{\ell}$ and $|\mathcal{I}_{x,j}^{\ell,m}|=N_{j}^{m}$ with $x\in\{s,d\}$. In order to reduce the signaling overhead between the UE and the gNB, we assume that the sets $\mathcal{I}_{s,j}^{\ell,m}$ and $\mathcal{I}_{d,j}^{\ell,m}$ are symmetric across all RBs in the entire RG, 
Hence, for layer $\ell$, the UE signals $\mathcal{P}_{1}^{\ell,1}\stackrel{\Delta}{=}(\mathcal{I}_{s,1}^{\ell,1},\mathcal{I}_{d,1}^{\ell,1})$ to the gNB which casts the information required by the gNB's transmission scheme. 
Due to the symmetry assumption, we drop the RB and SG dependence, i.e, the indices $j$ and $m$, in the per layer RB (PL-RB) pattern  $\mathcal{P}_{j}^{\ell,m}=(\mathcal{I}_{s,j}^{\ell,m},\mathcal{I}_{d,j}^{\ell,m})$ and $N_{j}^{m}$ and refer to them with $\Bar{\Bar{\mathcal{P}}}^{\ell}$ and $\Bar{\Bar{N}}$ respectively.

When a RG is divided into SGs, the size of the TBC blocks per SG, denoted by $\Bar{N}\stackrel{\Delta}{=}\Bar{\Bar{N}}N_{\mathrm{BSG}}$ and referred to as the view length, decreases. Hence, despite the sub-gridding role in decreasing the combining approximation error (through reducing the number of REs equalized with the same combiner), it motivates the importance of investigating the effect of the view length on the performance. In the numerical results section, we provide a detailed study on the effect of the view length $\Bar{N}$ and the SG size $N_{\mathrm{BSG}}$ on the system performance. 

\begin{figure}[t]
\begin{center}
\includegraphics[scale=.2,width=0.9\columnwidth]{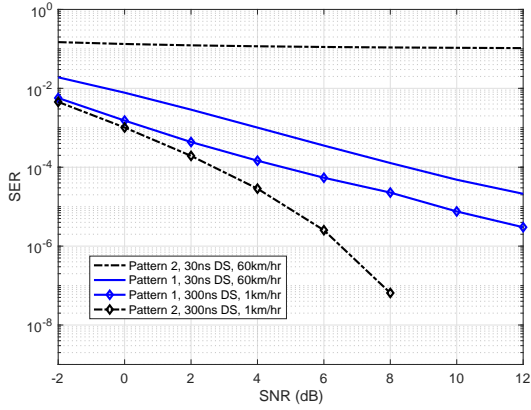}
\caption{SER vs SNR for a RG of 50 RBs, and $N_{\mathrm{BSG}}=2$, CDL-C channel model with SCS=30kHz.} 
\label{fig:pattdes}
\end{center}
\end{figure}
\section{Simulation Results}\label{Simulations}
We consider a DL transmission where a gNB, with $N_{\mathrm{t}}=8$ antennas, transmits QPSK data symbols to a UE, with $N_{\mathrm{r}}=2$ antennas, through a CDL-C channel with 15 kHz sub-carrier spacing (SCS) and a carrier frequency of 4 GHz. Data is transmitted into the form of 10 ms frames, where each frame is composed of 10 slots; each is 1 ms. We average the simulation results over 100 different channel seeds and 100 frames per seed point. 

We  first simulate the effect of different CCA patterns on the average SER. Figure \ref{fig:CCApatts} provides different CCA patterns for SGs of $N_{\mathrm{BSG}}=2$ RBs each. On one side, pattern 1 in \ref{fig:CCApatt1} is concentrated in one symbol per view, spans the entire frequency domain and time repetition is employed. On the other side, pattern 2 in Figure  \ref{fig:CCApatt3} spans the entire time domain, located in only two sub-carriers per view and is repeated in frequency. We simulate the average SER of the 2 patterns for a CDL-C channel with SCS = 30kHz and a RG of 50 RBs that is divided into SGs of size $N_{\mathrm{BSG}}=2$ RBs each. We consider two different channel setups with: 1) CDL-C channel of 30ns delay spread (DS) and 60km/hour UE speed. It can be seen, from figure \ref{fig:pattdes}, that pattern 2 outperforms pattern 1. 2) CDL-C channel of 300ns DS and 1km/hour UE speed. In the second  setup, pattern 2 completely fails to detect the transmitted symbols while pattern 1 performs better. The simulation results from figure \ref{fig:pattdes} motivate that a CCA pattern must be designed such that the symbols within a pattern should experience the same channel, i.e., symbols must be within the same TBC block.

\begin{figure}[t]
\begin{center}
\includegraphics[scale=.2,width=0.9\columnwidth]{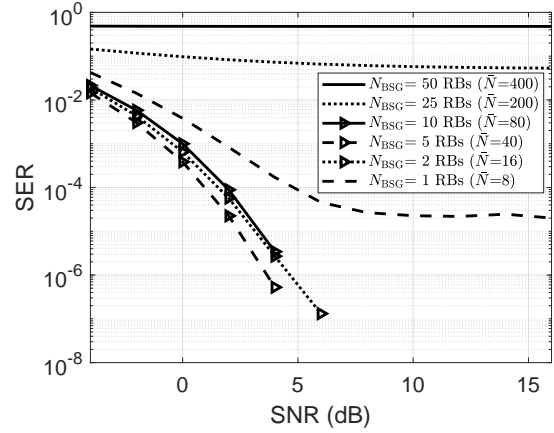}
\caption{SER vs SNR for for a RG of 50 RBs, SCS=30kHz and $N_{\mathrm{p}}=1$ symbols per SG.}  
\label{fig:SGSres1}
\end{center}
\end{figure}
We then study the impact of the SG size $N_{\mathrm{BSG}}$ on the SER performance. We consider a RG of $N_{\mathrm{RB}}=50$ RBs with SCS equal to 30 kHz, delay spread equal to 30 ns, and a PL-RB pattern $\Bar{\Bar{\mathcal{P}}}$ with $\Bar{\Bar{N}}=8$ REs. These choices ensure that the channel's coherence BW is less than the transmission BW leading to a frequency selective channel. In figure \ref{fig:SGSres1}, we simulate the average SER for different values of $N_{\mathrm{BSG}}$. It can be realized that when the SG size is large, $N_{\mathrm{BSG}}\in\{25,50\}$, CCA transmission provides a degraded average SER. This is because, due to the frequency selectivity nature of the channel, the symbols within a CCA signal experience different channel responses which violates the requirement that they should be constructed, and repeated, within a TBC block. Upon decreasing the SG size, $N_{\mathrm{BSG}}\in\{10,5,2\}$, the SER decreases as the channel affecting the CCA signal within each SG tends to be flatter. For a very small SG size,  $N_{\mathrm{BSG}}=1$, the average SER degrades again because the view length, $\Bar{N}$, per SG tends to be significantly small.  

\begin{figure}[t]
\begin{center}
\includegraphics[scale=.2,width=0.9\columnwidth]{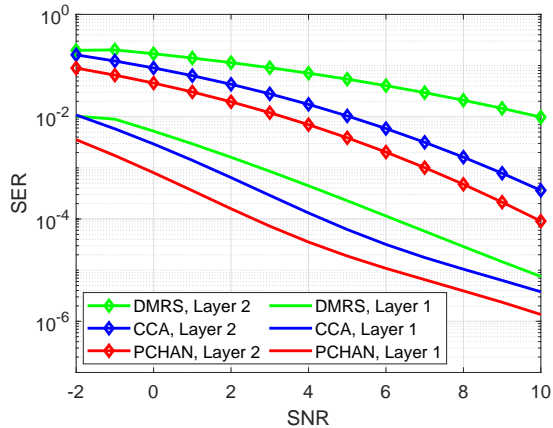}
\caption{SER vs SNR for a RG of 52 RBs, DS=30 ns \& UE speed=5km/hour, $N_{\mathrm{BSG}}=4$ and an $8\times4$ antennas system.}  
\label{fig:MLTX}
\end{center}
\end{figure}
Next we consider a DL transmission scenario where a gNB transmits two data layers, i.e., $L=2$, to the UE. A CCA pattern with $\Bar{\Bar{N}}=8$ REs is used in the first layer. In order to ensure that the CCA recovers one common signal at each layer, we use a shifted version of the first layer's pattern in the second layer's transmission; this way the patterns do not overlap. In figure \ref{fig:MLTX}, we present the average SER for each of the layers separately and compare the CCA equalization performance with a DM-RS pattern of mapping type 1, configuration type 2, length 1 and 1 additional position as defined in \cite{3gpp38211,3gpp38212,3gpp38213,3gpp38214}. Perfect channel knowledge (PCHAN) results, are incorporated as a lower bound on the performance. It can be seen that the considered CCA pattern outperforms the DM-RS one in the simulated channel conditions. Moreover, the gap between CCA and DM-RS in layer 1 is less than that of layer 2. This is because the interference power on layer 2 is larger than that of layer 1, assuming equal power allocation across layers, which leads to a degraded SER for DM-RS as an inaccurate channel estimate is used for equalization. However, CCA is not affected as it tends to recover the common signal regardless of the interference power.

To show how robust CCA is to strong interference, we now treat the second layer as an unknown interference. This can be used to model a multi-user (MU) MIMO network with two gNBs; each transmits data to an assigned UE. For either of the UEs in that network, the interference is then due to the transmission of the secondary gNB to its intended UE.  We set the first layer power to $\alpha_1$, while the second layer power is set to $1-\alpha_1$. The signal to interference ratio (SIR) then can easily found to be $\text{SIR}=\frac{\alpha_1}{1-\alpha_1}$. We vary the SIR value in our simulation while satisfying a total power budget constraint. In order to ensure that the SNR remains fixed for the first layer, the noise power is scaled with a factor of $\alpha_1$. In figure \ref{fig:interference}, we simulate the average SER vs SIR for a CCA pattern with $\Bar{\Bar{N}}=8$ REs and compare it with the same DM-RS pattern in the previous experiment. On one side, CCA equalization is not affected by the low SIR values. This is because CCA aims to recover the common signal regardless of the interference power. However, on the other side, DM-RS is highly influenced by the interference power. It fails to calculate an accurate effective channel estimate to be used further for equalization and therefore, it provides a high SER in the interference limited (low SIR) range. On increasing SIR, the channel switches to being a noise limited one where the performance converges to the average SER achieved by the simulated noise range, i.e., convergence is to zero for noiseless case and $2\times 10^{-3}$ average SER for 2 dB noise in the DM-RS case.

\begin{figure}[t]
\begin{center}
\includegraphics[scale=.2,width=0.9\columnwidth]{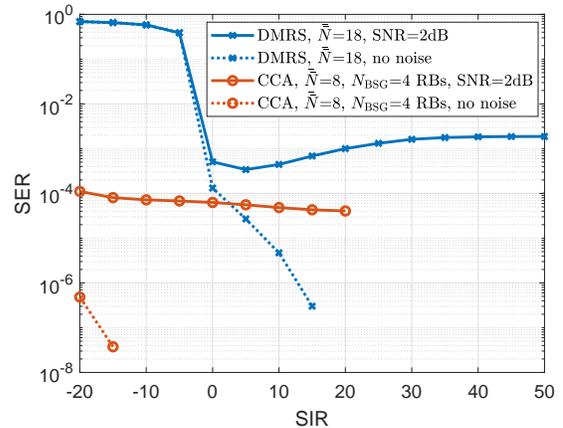}
\caption{SER vs SIR for a RG of 52 RBs, DS=30 ns, UE speed=5km/hour.}
\label{fig:interference}
\end{center}
\end{figure}

\section{Conclusions} \label{Conc} 
We considered the problem of RS overhead reduction through proposing a new data structure in a downlink single user MIMO system. The new data structure requires repetition of a portion of the PDSCH data across all sub-band. Utilizing the repetition structure at the UE side, we showed that reliable detection of the desired signal is possible via CCA. The proposed solution is shown to provide both performance and computational gains, rendering it practically feasible. To further boost the CCA performance, we presented two effective strategies to select the repetition pattern type and to deal with the high frequency selective channel scenarios. We showed through numerical results that both strategies are effective in improving the CCA performance. We demonstrated through extensive simulations on a 3GPP link-level testbench that the proposed approach considerably outperforms the state-of-the-art methods, while having a lower computational complexity.  
%\end{sloppypar}

\bibliographystyle{IEEEtran}
\bibliography{IEEEabrv,Ref}
\end{document}